\documentclass[letterpaper, 10 pt, conference]{ieeeconf}

\IEEEoverridecommandlockouts   

\usepackage{amsmath,amssymb,amsfonts}
\usepackage{mathtools}
\usepackage{bm}
\usepackage{graphicx}
\usepackage{xcolor}
\usepackage[colorlinks=true,allcolors=blue]{hyperref}
\usepackage{cite}
\usepackage{comment}
\usepackage{dsfont}
\usepackage{booktabs}
\usepackage{mathrsfs}
\newtheorem{definition}{Definition}
\newtheorem{proposition}{Proposition}
\newtheorem{assumption}{Assumption}
\newtheorem{theorem}{Theorem}
\newtheorem{remark}{Remark}
\newtheorem{lem}{Lemma}

\renewcommand{\thesection}{\Roman{section}}
\renewcommand{\thesubsection}{\Alph{subsection}}

\newcommand{\loc}{L}                 % discrete mode counter
\newcommand{\cX}{\mathcal{X}}           % continuous state space
\newcommand{\cL}{\mathcal{L}}           % infinitesimal generator
\newcommand{\R}{\mathbb{R}}             % reals
\newcommand{\Prob}{\mathbb{P}}          % probability of the FSP event
\newcommand{\UST}{\mathrm{UST}}         % the sequential-safety predicate
\newcommand{\cF}{\mathcal{F}}
\DeclareMathOperator{\tr}{tr}

\newcommand{\trajj}{z_{t_0}^{T}}

\newcommand{\mnl}[1]{{\color{blue} Manuela: #1}}
\newcommand{\abhi}[1]{{\color{violet} Abhijit: #1}}

\title{\LARGE\bf Probabilistic Sequential Safety for Stochastic Dynamical Systems}

\author{Abhijit Mazumdar, Manuela L. Bujorianu and Rafal Wisniewski 
\thanks{Authors are with the Section of Automation $\&$ Control, Aalborg University, 9220 Aalborg East, Denmark (e-mail: \{abma, lmbu, raf\}@es.aau.dk).}
\thanks{This work has been supported by the Independent Research Fund Denmark (DFF) 
        under Project SAFEmig (Project Number 3105-00173B).}
}

\begin{document}
\maketitle
\thispagestyle{empty}
\pagestyle{empty}

\begin{abstract}
In this letter, we investigate a notion of probabilistic safety for a stochastic dynamical system. Standard safety notions are based on the idea that an unsafe region of the state-space must be avoided at each instant. In contrast, we study the \textbf{notion of sequential safety} which is based on avoiding a tuple of target regions of the state-space. The proposed notion is motivated by applications ranging from classical control systems, such as a vehicle that exceeds a speed limit before reaching an obstacle, to modelling multi-turn jailbreak attacks on large language models (LLMs) and predicting hazardous chemical reactions.
%The proposed safety notion is inspired by many practical systems, ranging from modeling jailbreak attacks in large language models (LLMs) to predicting hazardous chemical reactions. %In solving the sequential safety problem, we solve a \textbf{sequential reachability problem} for a stochastic dynamical system. 
\end{abstract}


\section{Introduction}

Safety verification of stochastic dynamical systems is a fundamental problem in
control theory~\cite{ames2019control}. Traditional
safety notions concern a single distinguished region: trajectories must reach a
target while avoiding an unsafe
set~\cite{prajna2007framework, wisniewski2021safety, wisniewski2017stochastic}.
Many practical specifications, however, involve several regions and depend on
the pattern in which they are visited: either all of them are visited before a
deadline in \emph{any} order, or in a \emph{prescribed} one. A primary
motivating example is a large language model (LLM) under a multi-turn jailbreak
attack, where a sequence of individually innocuous prompts, each passing
conventional safety filters, progressively weakens the model's guardrails
before the final harmful request~\cite{russinovich2025great}. 
Similarly, in chemical process safety, certain reaction pathways are hazardous
only when temperature and pressure setpoints are crossed in an adverse order.
Instantaneous safety metrics, which treat each moment independently, cannot
capture such temporal and ordinal structure.

\textit{Related Literature:}
%Several works have developed probabilistic {single-target} reachability methods specifically for stochastic differential equations and diffusion processes.
Stochastic reach-avoid problems for diffusions have been studied from an optimal-control perspective in \cite{esfahani2016stochastic}. A probabilistic safety notion called \(p\)-safety is introduced in ~\cite{wisniewski2020p} for stochastic hybrid processes, while \cite{wisniewski2023probabilistic} develops a probabilistic safety framework for Markov processes. Barrier-function and finite-abstraction approaches have provided alternative ways of obtaining probabilistic safety guarantees for stochastic systems~\cite{laurenti2025unifying}. In contrast to these certificate- and abstraction-based methods, our objective is to obtain the sequential reachability probabilities.
%They established a dynamic programming characterization of reach-avoid probabilities and derived corresponding PDE characterizations in terms of discontinuous viscosity solutions. Their formulation considers the probability of reaching a target before an avoid set, either before a prescribed horizon or at the terminal time. This work is particularly close to the present setting in that it connects probabilistic reachability of diffusions with boundary-value PDEs. Probabilistic safety has also been formulated for general Markov processes. 

Closely related is the literature on stochastic hybrid systems, where reachability has been studied for single targets \cite{bujorianu2008reachability}, under state constraints \cite{bujorianu2011state}, and, in discrete time, as finite-horizon reachability and safety via dynamic programming~\cite{abate2008probabilistic,summers2010verification}.

In probabilistic model checking, multiphase-until formulas express sequential requirements and have been studied for continuous-time Markov chains and Markov reward models~\cite{xu2016multiphase}. Unlike these finite-state settings, our state space is continuous and the dynamics are given by an Itô diffusion.
\textit{Our Contribution:} Existing work on stochastic safety of diffusions 
primarily addresses single-target reachability, reach-avoid, or safety specifications. This letter develops a framework for \emph{probabilistic sequential safety} for 
Itô diffusions. Our specific contributions are:

\begin{enumerate}
  \item We introduce \emph{probabilistic sequential safety}, which characterizes 
        the probability that a stochastic process visits a tuple of target 
        regions $(U_1, \ldots, U_N)$ either in a prescribed sequence or in any order 
        before a horizon $T$.

  \item \emph{Ordered case:} We show that the probability of visiting the target 
        regions in a prescribed order is given by a chain of recursive coupled backward 
        Kolmogorov equations, one for each stage of the sequence.  The derivation relies on the strong Markov property and the 
        Feynman--Kac representation~\cite{lawler2010stochastic}, and gives a PDE-based counterpart 
        to existing sequential probabilistic verification methods.  

  \item \emph{Unordered case:} We extend the framework to the setting where all 
        target regions must be visited in any order. These are solved 
        backward, starting from the case where every region has been visited. 
        The number of equations grows exponentially, so this case costs more than 
        the ordered one. It still scales far better than the factorial growth of 
        enumerating every possible visiting order.
\end{enumerate}

The remainder of this letter is organized as follows. Section~II introduces the stochastic system and the standing assumptions. Section~III addresses the ordered case and Section~IV extends this characterization to the unordered case via a visit-set augmentation. Section~V illustrates the results on a longitudinal vehicle example, validating the PDE solution against Monte Carlo simulation, and Section~VI concludes with directions for future work.

%The remainder of this letter is organized as follows. In Section~II, we introduce the formal problem setup, including the mode-augmented Markov process and the definition of sequential-passage safety. In Section~III, we state the main theorem characterizing the safety probability via coupled backward Kolmogorov equations. In Section~IV, we discuss the numerical approach, and Section~V concludes with remarks on future directions.

\textit{Notations:} %We use $\mathds{R}^n$ to denote the $n$-dimensional Euclidean space and $\mathds{R}^{n\times m}$ to denote the set of $n\times m$ real matrices. 
For a vector $h\in \mathds{R}^n$, we use $h^T$ to denote its transpose. For an open set $\cX \subset \mathds{R}^n$, $C^2(\cX)$ represents twice continuously differentiable functions on $\cX$.  For a subset $A$ of a topological space, we use $\partial A$ to denote the boundary of $A$. For a square matrix $M$, we represent the trace by $\tr(M)$. $\mathbf{1}_{A}$ is used to denote the indicator function of event $A$; equals $1$ when the event is realized, $0$ otherwise. For a set $S$, $\mathcal{P}(S)$ denotes the power set of $S$. For $a,b\in\R\cup\{+\infty\}$, $a\wedge b$ denotes their minimum.

\section{System Descriptions}
\label{sec:problem}

%\subsection{Stochastic Dynamics}

Let $(\Omega,\mathcal{F},\mathds{P})$ be a 
probability space carrying a standard $m$-dimensional Brownian motion
$W$. We consider the time-homogeneous It\^o diffusion process in a state-space
$\cX \subseteq \R^n$,
\begin{equation}
    \mathrm{d}X_t = f\bigl(X_t\bigr)\,\mathrm{d}t
    + g\bigl(X_t\bigr)\,\mathrm{d}W_t,
    \qquad X_0=x_0,
    \label{eq:sde}
\end{equation}
with drift $f:\cX\to\R^n$ and diffusion coefficient
$g:\cX\to\R^{n\times m}$. We write $(X_t)_{t\geq 0}$ to represent a process. We use $\mathds{P}^{Z}$ to denote the probability measure conditioned on a random variable $Z$ and $\mathds{E}^{Z}$ to denote the expectation with respect to $\mathds{P}^{Z}$. For any process $(X_t)_{t\geq 0}$, we use uppercase $X_t$ to denote the random variable at time instant $t$, and $x_t$ is used to represent its realization.  Let $\mathcal{F}_t=\sigma(X_s:s\leq t)$ be the natural filtration of
$(X_t)_{t\geq 0}$. For a random variable $\tau$ and a constant $T$, we write
$\{\tau \le T\} := \{\omega \in \Omega : \tau(\omega) \le T\}$.
%We write $P^{x}$ for the probability kernel and expectation of the process conditioned on $X_t=x$. 
%We write $P_{t,x}$ and $E_{t,x}$ for the probability kernel and expectation of the process conditioned on $x(t)=x$, and abbreviate $P_x := P_{t,x}$.

We make the following standing assumption throughout the paper.
%\begin{assumption}\label{ass:wellposed}
%$f$ and $g$ are locally Lipschitz and satisfy a linear-growth bound, so that~\eqref{eq:sde} admits a unique strong solution with continuous sample paths for all $t\ge t_0$.
%\end{assumption}
\begin{assumption}
\label{ass:wellposed}
$f$ and $g$ are locally Lipschitz and satisfy a linear-growth bound,
so that \eqref{eq:sde} admits a unique strong solution in the sense 
of~\cite{mao2007stochastic}, with continuous sample paths. Moreover, the
diffusion is uniformly elliptic: there is $\lambda>0$ such that
$\xi^\top g(x)g(x)^\top\xi \ge \lambda\|\xi\|^2$ for all $x\in\mathcal{X}$, $\xi\in\mathbb{R}^n$.
\end{assumption}

%\textit{Assumption 2:} The target sets $U_1,\dots,U_N$ are closed, have nonempty interior, and each finite union $\bigcup_{k\in K}U_k$, $K\subseteq\mathcal{N}$, satisfies the interior cone condition; equivalently, its complement satisfies the exterior cone condition, so that every boundary point is regular for (1).}
%\subsection{Sequential Reachability}
\begin{comment}
We are interested in a tuple of target regions
$U_1,\ldots,U_N \subseteq \cX$, each \textit{closed with a nonempty interior}, an
initial set $\cX_0\subseteq\cX$, and a finite time horizon $T$. Define the \emph{ordered hitting times}
\begin{equation}
    \tau_0 = t_0,\qquad
    \tau_k = \inf\bigl\{\, t > \tau_{k-1} : X_t\in U_k \,\bigr\},
    \quad k=1,\ldots,N,
    \label{eq:guard}
\end{equation}
By construction $U_k$ is
credited only after $U_{k-1}$ has been visited, so
$\tau_0<\tau_1<\cdots<\tau_N\leq T$. The associated mode-counter
$\loc_t\in \mathcal{N} := \{0,1,\ldots,N\}$ records how many guards have been
crossed in order,
\begin{equation}
    \loc_t = \max\bigl\{\, k\in \mathcal{N} : \tau_k \le t \,\bigr\},
    \qquad \loc(t_0)=0,
    \label{eq:mode-dyn}
\end{equation}
i.e. $\loc_t$ is nondecreasing, piecewise constant, and increments by one
exactly when the next region $U_{\loc+1}$ in the prescribed order is
entered. We use $(\loc_t)$ to denote the process generated by \eqref{eq:mode-dyn}.
\end{comment}

We are interested in a tuple of target regions $U_1,\ldots,U_N \subseteq \cX$, an initial set $\cX_0:= \cX \setminus \cup_{i\in \mathcal{N} } U_i$, 
and a finite time horizon $T$. Suppose $\mathcal{N}:=\{1,\ldots,N\}$ denotes the index set of the target regions and $\mathcal{N}'$ denotes the set $\{0,1,\ldots,N\}$.  %Further, we assume that $U_1,\ldots,U_N \subseteq \cX$ are disjoint.

\begin{assumption}
\label{closed_targets_reg}
The target sets $U_1, \ldots, U_N$ are closed sets with nonempty interiors
and Lipschitz boundaries.
\end{assumption}

For the original diffusion process \eqref{eq:sde}, the infinitesimal generator $\cL$
acting on $C^2$ functions $u : \cX \to \R$ is
\small
\begin{equation}
    (\cL u)(x) = f(x)^\top \nabla u(x)
    + \tfrac{1}{2}\,\tr\!\bigl(g(x) g(x)^\top \nabla^2 u(x)\bigr).
    \label{generator}
\end{equation}
\normalsize

%\section{Main Results}
%\label{sec:diffusion-ordered-safety}

%In this Section, we provide our main results. First we consider the ordered safety problem, and then we provide results for the unordered sequential safety problem. 
%\subsection{Per-Mode Value Functions}
\section{Ordered Sequential Safety}
\label{sec:diffusion-ordered-safety}
In this section, we consider the ordered sequential safety specification. The ordered case is more straightforward than the unordered case. We first develop the main solution concept in the ordered setting and then extend it to the unordered one.

Define the \emph{ordered hitting times} as
\small
\begin{equation}
    \tau_0 = 0,\
    \tau_i = \inf\bigl\{\, t \geq \tau_{i-1} : X_t \in U_i \,\bigr\},
    \ i \in \mathcal{N}.
    \label{eq:guard}
\end{equation}
\normalsize
By construction, $U_i$ is credited only after $U_{i-1}$ has been visited, 
ensuring the ordering $\tau_0 \leq \tau_1 \leq \cdots \leq \tau_N \leq T$.

 The associated \emph{discrete mode} $\loc_t$ 
records how many target regions have been reached in the prescribed order:
\begin{equation}
    \loc_t = \max\bigl\{\, i \in \mathcal{N}' : \tau_i \le t \,\bigr\},
    \qquad \loc_{0} = 0.
    \label{eq:mode-dyn}
\end{equation}
%Equivalently, for all $t\in [0,T]$, $ s \in [0, t], \ i \in \mathcal{N}:$
%\small
%\begin{equation}
%\begin{split}
 %   \loc_t  = & i \iff X_{\tau_i} \in U_i \text{ and } X_s \notin U_{i+1} \ \iff \tau_i \leq t < \tau _{i+1}.
 %   \end{split}
 %   \label{eq:mode-membership}
%\end{equation}
%\normalsize

The process $(\loc_t)_{t \in [0, T]}$ is nondecreasing, piecewise constant, 
and increments by one at the stopping time $\tau_{i+1}$ when the next 
region $U_{i+1}$ is entered. Since $U_{i+1}$ is closed, the mode switches instantaneously upon entering $U_{i+1}$.

We define a mode-augmented process which combines both the diffusion process $(X_t)_{t\geq 0}$ and mode-counter process $(\loc_t)_{t\geq 0}$ as follows. %In the next section, we prove that the mode-augmented process is a Markov process. 

\begin{definition}[Mode-augmented process]
The process $(A_t)_{t\geq 0} := (X_t,\loc_t)_{t\geq 0}$ on $\cX\times \mathcal{N}'$ is called a mode-augmented process, where $(X_t)_{t\geq 0}$ and $(\loc_t)_{t\geq 0}$ are processes generated by \eqref{eq:sde} and \eqref{eq:mode-dyn}, respectively.
\end{definition}
Because the continuous dynamics~\eqref{eq:sde} are identical at every
mode $\loc_t \in \mathcal{N}'$, the
generator of the augmented process $(A_t)_{t\geq 0}$ acts mode-by-mode as a copy
of $\cL$. 
%coupled across modes only through the boundary conditions imposed by the unsafe regions. 

\begin{definition}[Probabilistic ordered sequential safety]\label{def:prob-safe}
For a given initial state $x_0$ and a finite time-horizon $T$, we call system \eqref{eq:sde} \emph{probabilistically ordered sequentially safe} with respect to a prescribed threshold $p \in [0,1]$ if
\begin{equation}\label{eq:safety-cond}
   % \Prob^{x_0}\bigl[\UST(\trajj)\bigr] \;\leq\; p.
  \Prob^{x_0}\bigl[ \tau_N \leq T\bigr] \;\leq\; p.
\end{equation}
\end{definition}

The verification problem addressed in the ordered sequential specification case is to compute the probability $\Prob^{x_0}\bigl[ \tau_N \leq T\bigr]$. %The next section reduces this to a coupled system of boundary value problems for the mode-augmented process $(A_t)_{t\geq 0}$.
\begin{comment}
\begin{remark}
    The practical significance of the sequential
specification is that it is strictly harder to satisfy than any of its constituent
requirements, and is therefore assigned a strictly smaller probability. Existing stochastic safety verification methods, which certify a single unsafe
region at a time or, equivalently, treat the collection $\bigcup_i U_i$ as one
unsafe set, therefore return upper bounds on the quantity of interest.
\end{remark}
\end{comment}
%\begin{remark}
%Every trajectory that completes the sequence also enters each target set individually. Methods that analyze each unsafe set separately, or treat their union as a single unsafe set, therefore only provide upper bounds on the sequential probability. %These bounds are exact only when the ordering is forced, for instance when one target set contains the other, and can be very conservative otherwise.
%\end{remark}

For each $i \in \mathcal{N}'$, define the value function
$v_i:[0,T]\times\cX\to[0,1]$ by
\small 
\begin{equation}
\begin{aligned}
     v_i(t,x) & := \mathds{P}\bigl[\loc_T = N \mid \loc_t=i,\ X_t=x \bigr] \\
    &= \mathds{P}\bigl[  \tau_{N}
       \leq T\mid \loc_t=i, \ X_t=x\bigr].
\end{aligned}
    \label{eq:vi-def}
\end{equation} 
\normalsize
The value function $v_i(t,x)$ is the probability
that, starting from state $x$ at time $t$ in mode $i$, the process
completes the remaining sequence of regions $U_{i+1}, \ldots, U_N$ in order
by time $T$. Since $\loc_t$ records the regions reached up to and including
time $t$, a realization started in mode $i$ at a point $x\in U_{i+1}$ switches
to mode $i+1$ at time $t$; hence $v_i(t,x)=v_{i+1}(t,x)$ for
$x\in U_{i+1}$. 

For $x_0 \in \cX_0$, the safety probability is:
\small
\begin{equation}
    \Prob^{x_0}[\tau_N\leq T] = \Prob \big[\loc_T=N\mid L_0=0, X_0=x_0\big] = v_0(0,x_{0}).
    \label{eq:fsp-as-v0}
\end{equation}
\normalsize
So the ordered sequential safety verification problem reduces to
computing $v_0$ on the initial set $\cX_0$ and checking
$v_0(0,x_0) \le p$ for all $x_0 \in \cX_0$.
\begin{comment}
\mnl{Here, my intuition says that we can introduce an auxiliary result (lemma) to give a representation $ v_i(t,x)$, which will be used in the proof of Th1 (you have it already in the proof). The most important value of this value function is when $x$ belongs to $U_i$, i.e. it corresponds to the hitting time. }
\end{comment}
\par Let $\mathcal{G}_t=\sigma(A_s:s\leq t)$ be the filtration of the augmented process. We relate $\mathcal{F}_t$ and $\mathcal{G}_t$ as follows. %Since $L_t$ is a deterministic functional of the realization of $(X_s)_{s\leq t}$, we have $\mathcal{F}_t\subseteq\mathcal{G}_t$. 
\begin{lem}\label{lem:filtrations}
For all $t\ge 0$, $\mathcal{G}_t=\mathcal{F}_t$.
\end{lem}
\begin{proof}
Since $\mathcal{G}_t$ is generated by $(X_s,L_s)_{s\le t}$ and $\mathcal{F}_t$
by $(X_s)_{s\le t}$ alone, $\mathcal{F}_t\subseteq\mathcal{G}_t$.

For the converse, it suffices to show that $L_t$ is $\mathcal{F}_t$-measurable
for every $t$. Since $\tau_1\le\tau_2\le\cdots\le\tau_N$ by \eqref{eq:guard} we have
$L_t=\sum_{i=1}^N\mathbf{1}_{\{\tau_i\le t\}}$, so it suffices that each
$\tau_i$ is an $\mathcal{F}_t$-stopping time. We argue by induction on $i$.
Clearly $\tau_0=0$ is a stopping time. Suppose $\tau_{i-1}$ is a stopping
time. By \eqref{eq:guard}, $\tau_i$ is the first entrance of $(X_t)_{t\geq 0}$ into $U_i$ after
$\tau_{i-1}$. Since $\tau_{i-1}$ is a stopping time, $U_i$ is closed and
$(X_t)_{t\geq 0}$ has continuous paths, whether this entrance has occurred before time
$t$ is determined by $(X_s)_{s\le t}$, i.e., $\{\tau_i<t\}\in\mathcal{F}_t$
for every $t$. By right-continuity of $\mathcal{F}_t$, also
$\{\tau_i\le t\}\in\mathcal{F}_t$, so $\tau_i$ is a stopping time. Thus $\mathcal{G}_t\subseteq \mathcal{F}_t$. 
%\par Therefore, $\mathcal{G}_t=\mathcal{F}_t$.
\end{proof}

\begin{comment}
\begin{proof}
Since $\mathcal{G}_t$ is generated by $(X_s,L_s)_{s\le t}$ and $\mathcal{F}_t$
by $(X_s)_{s\le t}$ alone, $\mathcal{F}_t\subseteq\mathcal{G}_t$.

For the converse, it suffices to show that $L_t$ is $\mathcal{F}_t$-measurable
for every $t$. By (3), $L_t=\sum_{i=1}^N\mathbf{1}_{\{\tau_i\le t\}}$, so it
suffices that each $\tau_i$ is an $\mathcal{F}_t$-stopping time. We argue by
induction on $i$. Clearly $\tau_0=t_0$ is a stopping time. Suppose
$\tau_{i-1}$ is a stopping time. By (2), $\tau_i$ is the first entrance of
$(X_t)$ into $U_i$ after $\tau_{i-1}$. Since $U_i$ is closed and $(X_t)$ has
continuous paths, whether this entrance has occurred before time $t$ is
determined by $(X_t)$ on $[t_0,t]$, so
$\tau_i$ is $\in\mathcal{F}_t$-measurable. Thus $\tau_i$ is a stopping time.

Hence $A_t=(X_t,L_t)$ is $\mathcal{F}_t$-measurable for every $t$, and
therefore $\mathcal{G}_t\subseteq\mathcal{F}_t$.
\end{proof}
\end{comment}

The following result establishes the strong Markov property of the mode-augmented process $(A_t)_{t\geq 0}$. The strong Markov property says the memoryless property holds not only at deterministic times $t$ but also at stopping times. We use it to derive our main result. %\mnl{Explain in words that strong MP means the Markov property wrt hitting times.} \author{This property is essential for all the results for Continuous Markov chain. Mention it.}

\begin{lem}[Mode-augmented Markov property]\label{stron_markov_lem}
The augmented process $(A_t)_{t\ge 0}=(X_t,L_t)_{t\ge 0}$ on
$\mathcal{X}\times\mathcal{N}'$ is a strong Markov process with respect to
$\mathcal{G}_t$.
\end{lem}

\begin{proof}
Let $\tau$ be a $\mathcal{G}_t$-stopping time. By
Lemma~\ref{lem:filtrations}, $\tau$ is also an $\mathcal{F}_t$-stopping time
and $\mathcal{G}_\tau=\mathcal{F}_\tau$. By the strong Markov property of the
process $(X_t)_{t\geq 0}$, on $\{\tau<\infty\}$ and conditional on $\mathcal{F}_\tau$,
the process $(X_{\tau+s})_{s\ge0}$ has the law of the diffusion
started from $X_\tau$.

The future mode-augmented process is obtained deterministically from $(X_{\tau+s})_{s\ge0}$ and the current mode $L_\tau$. Hence, on each event $\{L_\tau=i\}$,
$i\in\mathcal{N}'$, which belongs to $\mathcal{G}_\tau$, the future
$(A_{\tau+s})_{s\ge0}$ is a fixed measurable functional of
$(X_{\tau+s})_{s\ge0}$. Applying the strong Markov property of $(X_t)_{t\geq 0}$ on
each such event and summing over $i\in\mathcal{N}'$ gives, for every bounded
measurable $h:\mathcal{X}\times\mathcal{N}'\to\mathbb{R}$ and every $t\ge0$,
\[
\mathbb{E}\big[h(A_{\tau+t})\mid\mathcal{G}_\tau\big]
=\mathbb{E}^{A_\tau}\big[h(A_t)\big]\quad\text{on }\{\tau<\infty\}.
\]
Thus $(A_t)_{t\ge 0}$ is a strong Markov process.
\end{proof}

\begin{theorem}[Ordered Sequential Safety Probability]\label{thm:ordered}
Let Assumptions 1 and 2 hold with $\mathcal X=\mathbb R^n$, and let
$D_i:=\mathcal X\setminus U_{i+1}$. For $i=N-1,N-2,\dots,0$, the value
function $v_i(t,x)$ restricted to $[0,T]\times (D_i\cup\partial U_{i+1})$ is the
unique bounded solution, continuous on $[0,T]\times (D_i\cup\partial U_{i+1})$
except on $\{T\}\times\partial U_{i+1}$, of the following
terminal-boundary-value problem:
\small
\begin{equation}\label{eq:thm1}
\left\{
\begin{aligned}
&\frac{\partial v_i(t,x)}{\partial t}+\mathcal L v_i(t,x)=0,
  && \forall (t,x)\in[0,T)\times D_i,\\
&v_i(t,x)=v_{i+1}(t,x), && \forall x\in\partial U_{i+1},\ \forall t\in[0,T],\\
&v_i(T,x)=0, && \forall x\in D_i,\\
&v_N(t,x)=1, && \forall (t,x)\in[0,T]\times\mathcal X .
\end{aligned}
\right.
\end{equation}
\normalsize
Moreover, $v_i(t,x)=v_{i+1}(t,x)$ for all $x\in U_{i+1}$ and
$t\in[0,T]$, so that $v_{i+1}$ is defined on all of $\partial U_{i+1}$,
including $\partial U_{i+1}\cap U_{i+2}$. The ordered sequential safety
probability is given by
\[
\mathbb P^{x_0}\bigl(\tau_N\le T\bigr)=v_0(0,x_0),\qquad x_0\in\mathcal X_0.
\]
\end{theorem}

\begin{comment}
\begin{theorem}
\label{thm1}
The value function $v_i(t,x)$ for $x \in \cX \setminus \operatorname{cl}({U_{i+1}})$, $i = N-1, N-2, \ldots, 0$, is the solution of the linear terminal-boundary-value problem
\begin{equation}
    \begin{cases}
      \frac{\partial v_i(t,x)}{\partial t}   + \cL v_i(t,x) = 0,
            & (t,x) \in [t_0,T) \times (\cX\setminus U_{i+1}), \\
        v_i(t,x) = v_{i+1}(t,x), & x \in \partial U_{i+1},\ t\in[t_0,T], \\
        v_i(T,x) = 0, & x \notin U_{i+1}, \\
         v_N(t,x) = 1 &  \text{for all }
    (t,x) \in [t_0,T]\times\cX.
    \end{cases}
    \label{eq:pde-system}
\end{equation}
The sequential safety probability is given by
\begin{equation}
    \Prob^{x_0}\bigl[\tau_N\leq T\bigr] = v_0(t_0,x_0), \qquad x_0\in\cX_0.
\end{equation}
\end{theorem}
\end{comment}

\begin{proof}
For each mode $i\in\mathcal{N}'$, let
\begin{equation}
  v_i(t,x)=\Prob \bigl[\loc_T=N \,\big|\, \loc_t=i, \ X_t=x\bigr]
  \label{eq:vi}
\end{equation}
be the probability of completing the ordered crossings $U_{i+1},\dots,U_N$
by the horizon $T$, started from $(t,x)$ in mode $i$. Since $\loc_t$ records
the regions reached up to and including time $t$, a realization started in
mode $i$ at a point $x\in U_{i+1}$ switches to mode $i+1$ at time $t$.

If $\loc_t=N$ the specification is already met, so
$v_N(t,x)=1$ on $[0,T]\times\cX$. This is the last condition in the theorem.

Now consider $i \in \{0,1,\dots,N-1\}$. If $x\in U_{i+1}$, the mode switches
to $i+1$ at time $t$, so starting from $(t,x)$ in mode $i$ and in mode $i+1$
describe the same situation, and
\begin{equation}
  v_i(t,x)=v_{i+1}(t,x),\qquad \forall x\in U_{i+1},\ \forall t\in[0,T].
  \label{eq:switch}
\end{equation}
%This is the final assertion of the theorem. Relation \eqref{eq:switch} is stated on all of $U_{i+1}$ because the target regions need not be disjoint: a point of $\partial U_{i+1}$ may lie in $U_{i+2}$, where the stage-$i$ boundary datum $v_{i+1}$ is not determined by the PDE for $v_{i+1}$ but by \eqref{eq:switch} applied at stage $i+1$. 
In particular, \eqref{eq:switch}
holds on $\partial U_{i+1}$, which is the second line of \eqref{eq:thm1}.

Next let $x\notin U_{i+1}$. When $\loc_t=i$ the realization has not entered
$U_{i+1}$ after $\tau_i$, so $\tau_{i+1}$ is the first entrance of
$(X_s)_{s\ge t}$ into $U_{i+1}$. Since $U_{i+1}$ is closed and $(X_t)_{t\geq 0}$ has
continuous paths, $X_{\tau_{i+1}}\in\partial U_{i+1}$ on
$\{\tau_{i+1}<\infty\}$: a realization starting outside $U_{i+1}$ enters it
through its boundary. Since $\loc_T=N$ requires $U_{i+1}$ to be reached by
time $T$,
\small
\begin{align}
\begin{split}
  &v_i(t,x)
  = \mathds{E}\Bigl[\mathbf{1}_{\{\tau_{i+1} \le T\}}
     \cdot \mathds{P}\left[\loc_T=N \mid \mathcal{G}_{\tau_{i+1}} \right] \mid X_t=x, \loc_t=i\Bigr].
  \label{eq:decompose}
\end{split}
\end{align}
\normalsize
At $\tau_{i+1}$ the mode switches to $i+1$. If $X_{\tau_{i+1}}$ also lies in
$U_{i+2},\dots$, it switches further at the same instant, which by
\eqref{eq:switch} applied to the subsequent stages does not change the value.
Applying the strong Markov property of $(A_t)_{t\geq 0}$ from
Lemma~\ref{stron_markov_lem}, on $\{\tau_{i+1}\le T\}$ the conditional success
probability given $\mathcal{G}_{\tau_{i+1}}$ equals the success probability for
stage $i+1$ started at the restarted state $X_{\tau_{i+1}}$:
\begin{align}
  \mathds{P}\left[\loc_T=N \mid \mathcal{G}_{\tau_{i+1}}\right]&= v_{i+1}\bigl(\tau_{i+1}, X_{\tau_{i+1}}\bigr).
  \label{eq:markov-restart}
\end{align}
Substituting back into \eqref{eq:decompose}:
\small
\begin{equation}
  v_i(t,x)
  = \mathds{E} \bigl[\,v_{i+1}\bigl(\tau_{i+1}, X_{\tau_{i+1}}\bigr)\,
     \mathbf{1}_{\{\tau_{i+1} \le T\}}\, \mid X_t=x, \loc_t=i\bigr].
  \label{eq:onestep}
\end{equation}
\normalsize
Since $X_{\tau_{i+1}}\in\partial U_{i+1}$, \eqref{eq:onestep} uses the values
of $v_{i+1}$ on $\partial U_{i+1}$ only. This is why the boundary condition
in \eqref{eq:thm1} is stated on $\partial U_{i+1}$ rather than on all
of $U_{i+1}$.

If $x\in D_{i}$ and $t=T$, then $X_T=x$ lies outside the closed set
$U_{i+1}$, so $\tau_{i+1}>T$ and \eqref{eq:onestep} gives $v_i(T,x)=0$. This
is the third condition of \eqref{eq:thm1}.

%To obtain the PDE, let $D_i:=\cX\setminus U_{i+1}$. 
On $\{\loc_t=i\}$ with
$x\in D_i$, $\tau_{i+1}$ is the first exit time of $(X_s)_{s\ge t}$ from
$D_i$. Define the boundary data
\small
\begin{equation}
  \psi_i(s,y):=
  \begin{cases}
    v_{i+1}(s,y), & \forall s\in[0,T],\ \forall y\in\partial U_{i+1},\\[2pt]
    0,            & s=T,\ \forall y\in D_i .
  \end{cases}
  \label{eq:psi}
\end{equation}
\normalsize
On $\{\tau_{i+1}\le T\}$ the stopped state is $(\tau_{i+1},X_{\tau_{i+1}})$
with $X_{\tau_{i+1}}\in\partial U_{i+1}$. On $\{\tau_{i+1}>T\}$ it is
$(T,X_T)$ with $X_T\in D_i$, where the integrand of \eqref{eq:onestep}
vanishes. Hence \eqref{eq:onestep} can be written as the exit-time
expectation
\small
\begin{equation}
  v_i(t,x)=\mathds{E}\bigl[\psi_i\bigl(\tau_{i+1}\wedge T,\,
           X_{\tau_{i+1}\wedge T}\bigr)\mid X_t=x,\ \loc_t=i\bigr].
  \label{eq:exit-rep}
\end{equation}
\normalsize
For $x\in\partial U_{i+1}$, $\tau_{i+1}=t$ and \eqref{eq:exit-rep} reduces
to \eqref{eq:switch}, so \eqref{eq:exit-rep} holds on
$[0,T]\times(D_i\cup\partial U_{i+1})$. Under Assumption~\ref{ass:wellposed} the law of $X_T$ is absolutely
continuous for $t<T$, and $\partial U_{i+1}$ is Lebesgue-null since it is
Lipschitz by Assumption~\ref{closed_targets_reg}, so
$\Prob[\tau_{i+1}=T\mid X_t=x,\loc_t=i]
\le\Prob[X_T\in\partial U_{i+1}\mid X_t=x,\loc_t=i]=0$ for $t<T$.
The stopped state therefore almost surely avoids the corner
$\{T\}\times\partial U_{i+1}$, where $\psi_i$ may be discontinuous.
Since $\partial U_{i+1}$ is Lipschitz by
Assumption~\ref{closed_targets_reg}, it satisfies the cone condition, so
every point of $\partial U_{i+1}$ is regular for $D_i$. By the Feynman--Kac representation
for diffusions stopped at the exit from a domain
\cite[Theorem 4.3.1]{lawler2010stochastic}, the function \eqref{eq:exit-rep} is the unique
bounded solution, continuous on $[0,T]\times(D_i\cup\partial U_{i+1})$
except on $\{T\}\times\partial U_{i+1}$, of the terminal-boundary-value
problem
\small
\begin{equation}
\left\{
\begin{aligned}
  &\frac{\partial v_i}{\partial t}+\mathcal{L}v_i=0,
     && \forall (t,x)\in[0,T)\times D_i,\\
  &v_i(t,x)=v_{i+1}(t,x), && \forall t\in[0,T],\ \forall x\in\partial U_{i+1},\\
  &v_i(T,x)=0, && \forall x\in D_i,
\end{aligned}
\right.
\label{eq:fk}
\end{equation}
\normalsize
for $i=N-1,N-2,\dots,0$. The boundary and terminal data in \eqref{eq:fk} are
the values of $\psi_i$ in \eqref{eq:psi}, which coincide with the second and
third lines of \eqref{eq:thm1} derived above. The first line of
\eqref{eq:fk} is the first line of \eqref{eq:thm1}.

Starting from $v_N\equiv1$, problem \eqref{eq:fk} is solved for
$i=N-1,\dots,0$, each step using the already-computed $v_{i+1}$ as lateral
boundary data on $\partial U_{i+1}$, and $v_i$ is extended to $U_{i+1}$ by
\eqref{eq:switch}. The sequential safety probability is
$\Prob^{x_0}\bigl[\tau_N\leq T\bigr]=v_0(0,x_0)$ for every $x_0\in\cX_0$.
\end{proof}

\section{Unordered Sequential Safety}
\label{sec:diffusion-unordered-safety}

We now consider the alternative specification where a system must visit all
$N$ target regions in \emph{any order} (rather than a prescribed sequence) by
time $T$. %This problem requires tracking which regions have been visited, rather than the index of the next region.

%\begin{assumption}\label{disjoint_targets}
%The target sets $U_1,\ldots,U_N$ are pairwise disjoint.
%\end{assumption}

%Let $\mathscr{N}:=\{1,2,\cdots,N\}$ denotes the set of all possible modes. 
We define the visit set
\begin{equation}
  S_t=\big\{k\in \mathcal{N}:\ X_s\in U_k \text{ for some } s\in[0,t]\big\},
  \label{eq:visit_set}
\end{equation}
which records the target regions visited up to time $t$ in any order. The process
$(S_t)_{t\ge 0}$ is nondecreasing and piecewise constant, with
$S_{0}=\emptyset$ whenever $x_0\notin\bigcup_{k}U_k$. Since each $U_k$ is
closed (Assumption~\ref{closed_targets_reg}) and $(X_t)_{t\geq 0}$ has continuous paths,
$S_t$ changes exactly at the first entrance times into the regions $U_k$.  Let
$\hat{\mathcal{G}}_t=\sigma\big((X_s,S_s):\ 0\le s\le t\big)$ denote the natural
filtration of the set-augmented process $(S_t,X_t)_{t\geq 0}$.

\begin{definition}[Probabilistic unordered sequential safety]\label{def:prob-unordered-safe}
For a given initial state $x_0$ and a finite time-horizon $T$, we call system \eqref{eq:sde} \emph{probabilistically unordered sequentially safe} with respect to a prescribed threshold $p \in [0,1]$ if
\begin{equation}\label{eq:safety-cond_2}
   % \Prob^{x_0}\bigl[\UST(\trajj)\bigr] \;\leq\; p.
  \Prob^{x_0}\bigl[ S_T=\mathcal{N} \bigr] \;\leq\; p.
\end{equation}
\end{definition}
\noindent
The verification problem we investigate in this section is to compute the probability $\Prob^{x_0}\bigl[ S_T=\mathcal{N} \bigr]$.

  For each subset $H \in \mathcal{P}(\mathcal{N})$, define the value function
\begin{equation}
  v_H(t,x) = \mathbb{P}\left(S_T = \mathcal{N} \,\big|\, S_t = H, X_t = x\right),
  \label{eq:unordered_value}
\end{equation}
which gives the probability that all $N$ regions are visited by time $T$, given that the subset $H$ has been visited by time $t$ and the state is $x$.

\begin{theorem}[Unordered Sequential Safety Probability]\label{thm:unordered}
Let Assumptions~\ref{ass:wellposed} and~\ref{closed_targets_reg} hold with
$\cX=\mathbb R^n$. For each $H\in\mathcal P(\mathcal N)$ let
$D_H:=\cX\setminus\bigcup_{k\notin H}U_k$ and
$J_H(x):=\{j\in\mathcal N\setminus H: x\in U_j\}$. For each
$H\subsetneq\mathcal N$, the value function $v_H(t,x)$ restricted to
$[0,T]\times(D_H\cup\partial D_H)$ is the unique bounded solution,
continuous on $[0,T]\times(D_H\cup\partial D_H)$ except on
$\{T\}\times\partial D_H$, of the coupled backward parabolic system
\small
\begin{equation}\label{eq:thm2}
\left\{
\begin{aligned}
&\frac{\partial v_H(t,x)}{\partial t}+\mathcal L v_H(t,x)=0,
  && \forall (t,x)\in[0,T)\times D_H,\\
&v_H(t,x)=v_{H\cup J_H(x)}(t,x), && \forall x\in\partial D_H,\ \forall t\in[0,T],\\
&v_H(T,x)=0, && \forall x\in D_H,\\
&v_{\mathcal N}(t,x)=1, && \forall (t,x)\in[0,T]\times\cX .
\end{aligned}
\right.
\end{equation}
\normalsize
Moreover, $v_H(t,x)=v_{H\cup J_H(x)}(t,x)$ for all $x\notin D_H$ and
$t\in[0,T]$. The unordered sequential safety probability is given by
\small
\[
\Prob^{x_0}\bigl(S_T=\mathcal N\bigr)=v_{J_\emptyset(x_0)}(0,x_0),
\qquad x_0\in\cX_0 .
\]
\normalsize
\end{theorem}

\begin{proof}
By definition, $v_H(t,x) = \mathbb{P}(S_T = \mathcal{N} \mid S_t = H, X_t = x)$
is the probability that all $N$ regions are visited by time $T$, starting from
state $(t,x)$ with already-visited subset $H$. Since $S_t$ records the regions
visited up to and including time $t$, every region containing $x$ is credited
at time $t$.

\noindent
If the process has visited all regions, then the objective is met, so
$v_{\mathcal{N}}(t,x) = 1$ for all $(t,x) \in [0,T] \times \mathcal{X}$.

\noindent
Now let $H \subsetneq \mathcal{N}$. If $x \notin D_H$, then 
$J_H(x) \neq \emptyset$, and every region in $J_H(x)$ is entered at time $t$. 
Starting from $(t,x)$ with visited subset $H$ and starting from $(t,x)$ with 
visited subset $H \cup J_H(x)$ therefore describe the same situation. Hence
\small
\begin{equation}
v_H(t,x) = v_{H \cup J_H(x)}(t,x), \quad \forall t \in [0,T], 
\quad \forall x \in \mathcal{X} \setminus D_H.
\label{eq:unordered_bc}
\end{equation}
\normalsize
This is the final assertion of the theorem.
Since $J_H(x)$ is determined by $x$ alone, \eqref{eq:unordered_bc} assigns 
exactly one value at each point, even when $x$ lies in several overlapping 
regions. Moreover, $x \notin U_j$ for every $j \notin H \cup J_H(x)$, so 
$x \in D_{H \cup J_H(x)}$, and the right-hand side of \eqref{eq:unordered_bc} 
is evaluated inside the domain of $v_{H \cup J_H(x)}$. For the stage with 
visited subset $H$ itself, only $x \in \partial D_H$ matters, which gives the 
second line of \eqref{eq:thm2}.

\noindent
Let $ \hat{\tau} = \inf\Big\{s \geq t : X_s \in \bigcup_{k \in \mathcal{N}\setminus H} U_k\Big\}$ be the first hitting time of the union of unvisited regions, that is, the
first exit time of $(X_s)_{s\ge t}$ from $D_H$. Let $x \in D_H$. Under
Assumption~\ref{closed_targets_reg}, this union is closed, and since
$(X_t)_{t\geq 0}$ has continuous paths, $\hat{\tau} > t$, 
$X_s \in D_H$ for $s \in [t, \hat{\tau})$, and $X_{\hat{\tau}} \in \partial D_H$ 
on $\{\hat{\tau} < \infty\}$. In particular, $J_H(X_{\hat{\tau}}) \neq \emptyset$. 
No unvisited region is entered before $\hat{\tau}$, and exactly the regions in 
$J_H(X_{\hat{\tau}})$ are entered at $\hat{\tau}$, so
$  S_{\hat{\tau}} = H \cup J_H(X_{\hat{\tau}}) \quad \text{on } \{\hat{\tau} < \infty\}$. On $\{\hat{\tau} > T\}$ no unvisited region is entered, so 
$S_T = H \neq \mathcal{N}$.

\noindent
Similar to Lemma~\ref{stron_markov_lem}, the subset-augmented
process $(X_t, S_t)_{t \geq 0}$ is a strong Markov process. Conditioning on
$\hat{\mathcal{G}}_{\hat{\tau}}$ and applying the strong Markov property, we 
obtain, on $\{\hat{\tau} \leq T\}$, $  \mathbb{P}\big(S_T = \mathcal{N} \mid \hat{\mathcal{G}}_{\hat{\tau}}\big)
  = v_{H \cup J_H(X_{\hat{\tau}})}(\hat{\tau}, X_{\hat{\tau}})$.

Hence for $x \in D_H$,
\small
\begin{equation}
\begin{split}
 & v_H(t,x)  \\
  &= \mathbb{E}\left[\mathbf{1}_{\{\hat{\tau} \leq T\}} \cdot 
  \mathbb{P}\big(S_T = \mathcal{N} \mid \hat{\mathcal{G}}_{\hat{\tau}}\big)
   \mid X_t=x, S_t=H \right]\\
  &= \mathbb{E}\left[\mathbf{1}_{\{\hat{\tau} \leq T\}} \cdot
  v_{H \cup J_H(X_{\hat{\tau}})}(\hat{\tau}, X_{\hat{\tau}}) 
  \mid X_t=x, S_t=H \right].
  \label{eq:unordered_decomp}
  \end{split}
\end{equation}
\normalsize

\noindent
If $x \in D_H$ and $t = T$, then $X_T = x$ lies outside the closed set
$\bigcup_{k \in \mathcal{N}\setminus H} U_k$, so $\hat{\tau} > T$ and
\eqref{eq:unordered_decomp} gives $v_H(T,x) = 0$.

\noindent
To obtain the PDE, define the boundary data
\small
\begin{equation}
  \psi_H(s,y) :=
  \begin{cases}
    v_{H\cup J_H(y)}(s,y), & s\in[0,T],\ y\in\partial D_H,\\[2pt]
    0, & s=T,\ y\in D_H,
  \end{cases}
  \label{eq:psi_H}
\end{equation}
\normalsize
which is well defined because $J_H(y)$ is determined by $y$. On
$\{\hat{\tau}\le T\}$ the stopped state is $(\hat{\tau},X_{\hat{\tau}})$ with
$X_{\hat{\tau}}\in\partial D_H$. On $\{\hat{\tau}>T\}$ it is $(T,X_T)$ with 
$X_T\in D_H$, where the integrand of \eqref{eq:unordered_decomp} vanishes. 
Hence \eqref{eq:unordered_decomp} can be written as the exit-time expectation
\small
\begin{equation}
  v_H(t,x)=\mathbb{E}\bigl[\psi_H\bigl(\hat{\tau}\wedge T,\,
           X_{\hat{\tau}\wedge T}\bigr)\mid X_t=x,\ S_t=H\bigr].
  \label{eq:unordered_exit_rep}
\end{equation}
\normalsize
For $x\in\partial D_H$ we have $\hat\tau=t$, and
\eqref{eq:unordered_exit_rep} reduces to \eqref{eq:unordered_bc}, so
\eqref{eq:unordered_exit_rep} holds on $[0,T]\times(D_H\cup\partial D_H)$.

\noindent
The set $\bigcup_{k\notin H}U_k$ need not have a Lipschitz boundary, but it
has the two properties used in the proof of Theorem~\ref{thm:ordered}.
First, $\partial D_H\subseteq\bigcup_{k\notin H}\partial U_k$, which is
Lebesgue-null by Assumption~\ref{closed_targets_reg}. Hence, as in
Theorem~\ref{thm:ordered},
$\mathbb{P}[\hat\tau=T\mid X_t=x,\,S_t=H]
\le\mathbb{P}[X_T\in\partial D_H\mid X_t=x,\,S_t=H]=0$ for $t<T$, and the
stopped state almost surely avoids the corner $\{T\}\times\partial D_H$,
where $\psi_H$ may be discontinuous. Second, each $y\in\partial D_H$ lies in
some $U_k$, $k\notin H$, but not in its interior, so $y\in\partial U_k$. The
cone at $y$ provided by the Lipschitz boundary of $U_k$ lies in
$U_k\subseteq\mathcal X\setminus D_H$, so every point of $\partial D_H$ is
regular for $D_H$.

\noindent
Finally, the boundary data in \eqref{eq:psi_H} are continuous on
$[0,T)\times\partial D_H$. Arguing by backward induction on $|H|$, assume
that every $v_{H'}$ with $H'\supsetneq H$ is continuous on
$[0,T)\times\mathcal X$; this is trivial for $v_{\mathcal{ N}}\equiv1$. Since the $U_j$ are closed, nearby boundary points lie in no more
unvisited regions than $y$ does, and by (21) the corresponding value
functions agree at $y$; hence the data are continuous.
%If $y_n\to y$ in $\partial D_H$, then $J_H(y_n)\subseteq J_H(y)$ for large $n$, because the sets $U_j$ are closed. Applying \eqref{eq:unordered_bc} at level $H\cup J_H(y_n)$ gives $v_{H\cup J_H(y_n)}(t,y)=v_{H\cup J_H(y)}(t,y)$, and since there are finitely many subsets, $v_{H\cup J_H(y_n)}(t,y_n)\to v_{H\cup J_H(y)}(t,y)$.

\noindent
By the Feynman--Kac representation for diffusions stopped at the exit from 
a domain~\cite[Theorem~4.3.1]{lawler2010stochastic}, the function 
\eqref{eq:unordered_exit_rep} is the unique bounded solution, continuous on
$[0,T]\times(D_H\cup\partial D_H)$ except on $\{T\}\times\partial D_H$, of
the terminal--boundary-value problem
\small
\begin{equation}
\left\{
\begin{aligned}
  &\frac{\partial v_H}{\partial t}+\mathcal{L}v_H=0,
     && (t,x)\in[0,T)\times D_H,\\
  &v_H(t,x)=v_{H\cup J_H(x)}(t,x), && (t,x)\in[0,T]\times\partial D_H,\\
  &v_H(T,x)=0, && x\in D_H.
\end{aligned}
\right.
\label{eq:unordered_fk}
\end{equation}
\normalsize
The boundary and terminal data in \eqref{eq:unordered_fk} are the values of
$\psi_H$ in \eqref{eq:psi_H}, which coincide with the second and third lines 
of \eqref{eq:thm2}. The first line of \eqref{eq:unordered_fk} is the 
first line of \eqref{eq:thm2}. Since $v_H$ is continuous on
$[0,T)\times(D_H\cup\partial D_H)$ and coincides with the continuous
functions $v_{H\cup J_H(x)}$ off $D_H$ by \eqref{eq:unordered_bc}, $v_H$ is
continuous on $[0,T)\times\mathcal X$, which completes the induction step.

\noindent
Since the boundary data for $v_H$ involve only $v_{H \cup J_H(x)}$ with 
$J_H(x) \neq \emptyset$, which has strictly more visited regions, the system 
can be solved backward starting from $v_{\mathcal{N}} \equiv 1$.

\noindent
Finally, starting from the empty visited set at time $0$ and applying 
\eqref{eq:unordered_bc} with $H = \emptyset$,
$$\mathbb{P}^{x_0}[S_T=\mathcal{N}] = v_\emptyset(0, x_0) 
= v_{J_\emptyset(x_0)}(0, x_0), \quad x_0 \in \mathcal{X}_0,$$
where $J_\emptyset(x_0) = \emptyset$ whenever $x_0 \in D_\emptyset$.
\end{proof}

\begin{remark}
Theorem~\ref{thm:unordered} provides a systematic method to compute the unordered 
sequential-passage safety probability by solving a coupled system of $2^{N}-1$ PDEs, one 
for each subset $i \subseteq \{1, \ldots, N\}$. In contrast, a naive approach based on Theorem~\ref{thm:ordered} would independently 
solve PDEs for each of the $N!$ possible visit orderings. 
%Beyond the exponential versus factorial distinction, the coupled structure in Theorem~\ref{thm:unordered}, particularly the boundary conditions linking adjacent subsets, allows for potential computational leverage: solutions for smaller subsets can inform those for larger ones, enabling backward-induction-like solution strategies that are not available when solving ordered sequences independently. %Thus Theorem~\ref{thm:unordered} is not merely a theoretical characterization but also a computationally tractable alternative to naive enumeration of orderings.
\end{remark}

\section{Numerical Example}
\label{sec:numerical}

We now apply Theorem~\ref{thm:ordered} to a vehicle on a straight road
approaching a hazard zone, such as a pedestrian crossing. The specification
is that the vehicle exceeds the speed limit and only afterwards enters the
zone. Neither event alone is the concern; what matters is the order, since
overspeeding on the approach reduces the margin for reacting to the hazard.
%\subsection{State and Dynamics}
Let $D_t \in \mathbb{R}$ denote the vehicle position and $V_t \in \mathbb{R}$ its
velocity, with state $X_t = (D_t, V_t) \in \mathbb{R}^2$. The vehicle evolves
according to
\small
\begin{align}
  \mathrm{d}D_t &= V_t \, \mathrm{d}t + \sigma_d \, \mathrm{d}W^1_t,
                   \label{eq:vehicle_position}\\
  \mathrm{d}V_t &= (u_t - \alpha V_t) \, \mathrm{d}t + \sigma_v \, \mathrm{d}W^2_t,
                   \label{eq:vehicle_velocity}
\end{align}
\normalsize
where $W^1$ and $W^2$ are independent standard Brownian motions, $\alpha > 0$ is
the friction and damping coefficient aggregating rolling resistance and
aerodynamic drag, $\sigma_d, \sigma_v > 0$ are noise intensities, and $u_t$ is the
acceleration input.

The term $\sigma_v\,\mathrm{d}W^2_t$ represents disturbances acting on the
velocity, such as wind gusts, changes in road grade and actuator noise. The
term $\sigma_d\,\mathrm{d}W^1_t$ reflects that the distance travelled is not
exactly the integral of the measured velocity: wheel slip and odometry errors
accumulate over time, so \eqref{eq:vehicle_position} holds only approximately.

%Retaining this term is both physically faithful and, as Section~\ref{subsec:nondegeneracy} records, analytically convenient.

%\subsection{Velocity Tracking Controller}
We close the loop with the proportional velocity-tracking control law
\small
\begin{equation}
  u_t = K(V_{\mathrm{ref}} - V_t), \qquad K > 0,
  \label{eq:vehicle_controller}
\end{equation}
\normalsize
where $V_{\mathrm{ref}}$ is the reference velocity and $K$ the feedback gain.
Substituting \eqref{eq:vehicle_controller} into \eqref{eq:vehicle_velocity} gives
\small
\begin{equation}
  \mathrm{d}V_t = \big[K V_{\mathrm{ref}} - (K+\alpha) V_t\big] \mathrm{d}t
                  + \sigma_v \, \mathrm{d}W^2_t .
  \label{eq:vehicle_cl}
\end{equation}
\normalsize
Since $-(K+\alpha) < 0$, the noise-free closed loop dynamics converges exponentially to   $V_* = \frac{K}{K+\alpha} V_{\mathrm{ref}}$.
 
%and $V_t$ is an Ornstein--Uhlenbeck process with stationary standard deviation $\sigma_v / \sqrt{2(k+\alpha)}$. Both quantities matter for the design of the example: a speed-limit threshold is reachable only through noise excursions above $V_*$, and is realised on the horizon $[0,T]$ only if it sits within a few stationary standard deviations of $V_*$.

%\subsection{Non-degeneracy of the Diffusion}
%\label{subsec:nondegeneracy}
Writing \eqref{eq:vehicle_position} and \eqref{eq:vehicle_cl} in the form of \eqref{eq:sde} with
\scriptsize
\begin{equation}
X_t = \begin{bmatrix} D_t \\ V_t \end{bmatrix}, \
  f(d,v) = \begin{bmatrix} v \\ K V_{\mathrm{ref}} - (K+\alpha)v \end{bmatrix},
  \ 
  g(d,v) = \begin{bmatrix} \sigma_d & 0 \\ 0 & \sigma_v \end{bmatrix},
\end{equation}
\normalsize
the diffusion matrix is $a = g g^{\!\top} = \operatorname{diag}(\sigma_d^2, \sigma_v^2)$.
For $\sigma_d, \sigma_v > 0$,
 $ \xi^{\!\top} a \, \xi \; \ge \; \min\{\sigma_d^2, \sigma_v^2\} \, \|\xi\|^2
  \qquad \forall \xi \in \mathbb{R}^2$,
so the generator is uniformly elliptic on $\mathbb{R}^2$. %This contrasts with the kinematically exact model $\sigma_x = 0$, for which $a = \operatorname{diag}(0, \sigma_v^2)$ is singular and ellipticity holds only in the hypoelliptic sense, through the bracket $[\,\mu, \Sigma_v\,]$ spanning the position direction. 
%Under \eqref{eq:uniform_ellipticity} the Dirichlet problems generated by Theorem~\ref{thm1} admit unique classical solutions on each stage domain, and the Feynman--Kac representation holds classically rather than in the viscosity sense.

The generator of $X_t$ is
\small
\begin{equation}
  \mathcal{L}\phi = v \phi_d + \big[K V_{\mathrm{ref}} - (K+\alpha)v\big] \phi_v
      + \tfrac{\sigma_d^2}{2} \phi_{dd} + \tfrac{\sigma_v^2}{2} \phi_{vv}.
  \label{eq:vehicle_generator}
\end{equation}
\normalsize
%If $W^1$ and $W^2$ are correlated with coefficient $\rho$, then \eqref{eq:vehicle_generator} acquires the cross term $\rho \sigma_x \sigma_v \phi_{xv}$ and $a$ remains positive definite precisely when $|\rho| < 1$; we take $\rho = 0$ below.

%\subsection*{Sequential Safety Specification:}
We consider two target regions,
\small
\begin{align}
  U_1 &= \big\{(d,v) \in \mathbb{R}^2 : v \ge v_{\max}\big\},
        \label{eq:unsafe_1}\\
  U_2 &= \big\{(d,v) \in \mathbb{R}^2 : d_{\mathrm{obs,min}} \le d \le d_{\mathrm{obs,max}}\big\},
        \label{eq:unsafe_2}
\end{align}
\normalsize
a high-speed violation zone and a static obstacle zone. %with safe set $\mathcal{S} = \mathbb{R}^2 \setminus (U_1 \cup U_2)$. 
Both sets are closed, as required by Assumption~\ref{closed_targets_reg}.

The quantity of interest is the probability of the \emph{ordered} passage: the
vehicle violates the speed limit and only subsequently reaches the obstacle,
$  v_0(0,x_0) = \mathbb{P}^{x_0}\!\big(\tau_1 \le \tau_2 \le T\big)$.

By Theorem~\ref{thm:ordered} with $N = 2$, this is obtained from the two-stage recursion

\small
\begin{align}
  \partial_t v_1 + \mathcal{L} v_1 &= 0 \quad \text{on } [0,T) \times (\mathcal{X} \setminus U_2), \nonumber\\
  v_1(t,\cdot) &= 1 \ \text{on } U_2, \quad v_1(T,\cdot) = 0 \ \text{on } \mathcal{X} \setminus U_2, \label{eq:stage1}
\end{align}
\normalsize
followed by
\small
\begin{align}
  \partial_t v_0 + \mathcal{L} v_0 &= 0 \quad \text{on } [0,T) \times (\mathcal{X} \setminus U_1), \nonumber\\
  v_0(t,\cdot) &= v_1(t,\cdot) \ \text{on } U_1, \quad v_0(T,\cdot) = 0 \ \text{on } \mathcal{X} \setminus U_1. \label{eq:stage0}
\end{align}
\normalsize

\vspace{-0.2cm}
\textit{Numerical Method:}
%\vspace{-0.2cm}
\label{subsec:numerics}
 
%Each stage of the recursion is a linear parabolic Dirichlet problem for the generator \eqref{eq:vehicle_generator}, and we solve it by the standard locally consistent finite-difference (Markov chain approximation) scheme of Kushner and Dupuis \cite{kushner2001}. Transport terms use hybrid central/upwind differencing \cite{spalding1972}, so that the discrete generator has non-negative off-diagonal entries and zero row sums; the scheme is therefore monotone and the computed values remain in $[0,1]$. Time stepping is implicit Euler on the truncated domain $[-1,4] \times [-0.5,2]$ with reflecting outer edges, and the Dirichlet data are imposed on the system matrix rather than on the right-hand side alone.

%Each stage is a linear parabolic Dirichlet problem for \eqref{eq:vehicle_generator}, solved with the Markov chain approximation scheme of \cite{kushner2001} on $[-1,4]\times[-0.5,2]$. Hybrid central/upwind differencing \cite{spalding1972} gives a discrete generator with non-negative off-diagonal entries and zero row sums, so the scheme is monotone and the values remain in $[0,1]$. We use implicit Euler time stepping and reflecting outer edges, and impose the Dirichlet data through identity rows in the system matrix.

Each stage is a parabolic PDE with Dirichlet data, solved on the box
$[-1,4]\times[-0.5,2]$ by replacing the diffusion with a Markov chain on
a grid whose jump rates match its drift and noise \cite{kushner1990numerical}.
The rates are constructed to be non-negative, so each computed value is a
weighted average of its neighbours and therefore stays in $[0,1]$, as a
probability must. We integrate in time with the implicit Euler scheme, which is unconditionally stable and preserves the bounds $0 \le v_i \le 1$. Reflecting conditions are imposed on the outer edges of the computational box, and the Dirichlet data on the target sets are enforced by replacing the corresponding rows of the system matrix with identity rows.

Two features of Theorem~\ref{thm:ordered} shape the implementation. First, the
boundary values for stage~$0$ on $U_1$ are the stage-$1$ solution $v_1$,
which changes in time, so $v_1$ is stored at every time step rather than
only at the final one. Second, the matrix solved at each time step does not
change within a stage, because the dynamics, grid and step size are fixed;
it is therefore factorised once per stage, and each time step
reduces to a cheap forward--backward substitution.
 
%Only two points are specific to Theorem~\ref{thm1}. First, the stage-$0$ problem \eqref{eq:stage0} inherits \emph{time-dependent} boundary data, namely the space--time trace of $v_1$ on $U_1$, so the stage-$1$ solution is retained over the whole time grid rather than at the terminal slice alone. Second, both stages share the same generator and the same grid, so a single sparse LU factorization of $I - \Delta t\,\mathcal{L}_\Delta$ is reused in all time steps of each stage.
 
As an independent reference we estimate $v_0(0,x_0)$ by Euler--Maruyama
simulation. First-passage events are corrected for discrete-monitoring bias by
Brownian-bridge crossing probabilities \cite{gobet2000}, in the velocity
channel for $U_1$ and in the position channel for $U_2$.\\
%the latter is needed precisely because $\sigma_x > 0$ makes $X$ a diffusion rather than a $C^1$ path.
%\vspace{-0.3cm}
\textit{Results:}\\
%\vspace{-0.2cm}
Parameters are $K = 2$, $\alpha = 0.5$, $V_{\mathrm{ref}} = 1$, hence
$V_* = 0.8$; noise intensities $\sigma_v = 0.5$ and $\sigma_d = 0.3$, giving a
stationary velocity deviation of $0.224$; obstacle zone
$[d_{\mathrm{obs,min}}, d_{\mathrm{obs,max}}] = [1.5, 2.0]$; horizon $T = 4$; and
initial condition $x_0 = [ 0 \ 0.8 ]^T$, that is, the vehicle starts at the
closed-loop equilibrium speed.

%Table~\ref{tab:vehicle_convergence} reports the grid convergence of $v_0(0,z_0)$ for $v_{\max} = 1.2$. The value is stable to four significant figures across a fourfold refinement in each coordinate.

%% ===================================================================
%%  Table I (grid convergence) + accompanying text, corrected.
%%  Fixes: (i) N_x, N_v, N_t defined at first use;
%%         (ii) v_0(0,x_0) -> v_0(0,z_0) in caption and column header;
%%         (iii) "fourfold refinement in each coordinate" was inaccurate
%%               (N_x x3.97, N_v x2.49, N_t x2.67) -- replaced by the
%%               total node count.
%% ===================================================================

%Table~\ref{tab:vehicle_convergence} reports the grid convergence of $v_0(0,x_0)$ for $v_{\max} = 1.2$. Here $N_d$ and $N_v$ denote the number of grid points in the position and velocity directions, respectively, and $N_t$ the number of time steps, so that the mesh spacings are $h_d = (d_{\mathrm{hi}} - d_{\mathrm{lo}})/(N_d - 1)$ and $h_v = (v_{\mathrm{hi}} - v_{\mathrm{lo}})/(N_v - 1)$, and $\Delta t = T/N_t$. The admissible resolutions are constrained by the requirement that $v_{\max}$ and the obstacle edges lie exactly on grid nodes at every level, which forces $N_d \in \{101, 201, 301, 401, \dots\}$ and $N_v \in \{101, 126, 176, 251, \dots\}$; snapping the thresholds to the nearest node instead would perturb the specification itself between refinement levels.

Table~\ref{tab:vehicle_convergence} reports the grid convergence of
$v_0(0,x_0)$ for $v_{\max} = 1.2$. Here $N_d$ and $N_v$ denote the number of
grid points in the position and velocity directions, respectively, so that
the mesh spacings are $h_d = (d_{\mathrm{hi}} - d_{\mathrm{lo}})/(N_d - 1)$ and
$h_v = (v_{\mathrm{hi}} - v_{\mathrm{lo}})/(N_v - 1)$; time is discretised by
implicit Euler with $\Delta t = T/N_t$. Since implicit Euler is first order in
time, each entry is Richardson-extrapolated from $N_t = 2400$ and $N_t = 1200$,
so that the table isolates the spatial discretisation error; the time error at
$N_t = 2400$ is $8\times10^{-5}$ on every grid. The resolutions are chosen such
that $x_0$, $v_{\max}$ and the obstacle edges lie exactly on grid nodes at
every level, which requires $N_d = 10m+1$ and $N_v = 25w+1$ for integers
$m, w \ge 1$; snapping the thresholds to the nearest node instead would perturb
the specification itself between refinement levels. The values decrease
monotonically under refinement, and the $301\times176$ grid used in
Table~\ref{tab:vehicle_validation} differs from the finest level by
$1\times10^{-4}$.

%The computed value is stable to four significant figures as the node count $N_x N_v$ increases from $1.0 \times 10^4$ to $1.0 \times 10^5$.

\begin{comment}
\begin{table}[t]
\centering
\caption{Grid convergence of $v_0(0,x_0)$ for $v_{\max} = 1.2$. Thresholds lie
exactly on grid nodes at every level.}
\label{tab:vehicle_convergence}
\begin{tabular}{ccc}
\toprule
$N_d \times N_v$ & $N_t$ & $v_0(0,x_0)$ \\
\midrule
$101 \times 101$ & 150 & 0.60144 \\
$201 \times 126$ & 200 & 0.60138 \\
$301 \times 176$ & 300 & 0.60153 \\
$401 \times 251$ & 400 & 0.60160 \\
\bottomrule
\end{tabular}
\end{table}
\end{comment}
\small
\begin{table}[t]
\centering
\caption{Grid convergence of $v_0(0,x_0)$ for $v_{\max} = 1.2$. Thresholds lie exactly on grid nodes at every level.}
\label{tab:vehicle_convergence}
\begin{tabular}{cccc}
\toprule
$N_d \times N_v$ & $h_d$ & $h_v$ & $v_0(0,x_0)$ \\
\midrule
$101 \times 101$ & $1/20$ & $1/40$  & 0.60279 \\
$201 \times 126$ & $1/40$ & $1/50$  & 0.60239 \\
$301 \times 176$ & $1/60$ & $1/70$  & 0.60219 \\
$401 \times 251$ & $1/80$ & $1/100$ & 0.60209 \\
\bottomrule
\end{tabular}
\end{table}
\normalsize

%Table~\ref{tab:vehicle_validation} compares both stage values against Monte Carlo as the speed limit $v_{\max}$ is varied, which sweeps the stringency of the ordering requirement while leaving the obstacle specification untouched. The PDE values use the $301 \times 176$ grid with $N_t = 300$; the Monte Carlo estimates use $2 \times 10^5$ paths with $\Delta t = 10^{-3}$ and are reported with $95\%$ confidence half-widths. 

Table~\ref{tab:vehicle_validation} compares $v_0(0,x_0)$ with Monte Carlo as
$v_{\max}$ varies, which tightens the ordering requirement while leaving the
obstacle unchanged. The PDE uses the $301 \times 176$ grid with Richardson
extrapolation in time ($N_t = 2400/1200$); Monte Carlo uses $10^5$
Euler--Maruyama paths ($\Delta t = 10^{-3}$) with Brownian-bridge corrections
and $95\%$ confidence intervals. In most of the cases, the PDE value lies within the
confidence interval, with relative differences below $1\%$, although the
probability itself falls from $0.602$ to $0.041$.

\begin{comment}
\begin{table}[t]
\centering
\caption{Ordered sequential-passage probability $v_0(0,x_0) = \mathbb{P}^{x_0}(\tau_1 \le \tau_2 \le T)$
for the vehicle example: Theorem~\ref{thm1} against Monte Carlo, as a function of
the speed limit $v_{\max}$.}
\label{tab:vehicle_validation}
\begin{tabular}{cccc}
\toprule
$v_{\max}$ & Theorem~\ref{thm1} & Monte Carlo & rel.\ err. \\
\midrule
$1.2$ & 0.60153 & $0.60282 \pm 0.00214$ & $0.21\%$ \\
$1.3$ & 0.32287 & $0.32427 \pm 0.00205$ & $0.43\%$ \\
$1.4$ & 0.13082 & $0.13275 \pm 0.00149$ & $1.45\%$ \\
$1.5$ & 0.04085 & $0.04103 \pm 0.00071$ & $0.44\%$ \\
\bottomrule
\end{tabular}
\end{table}
\end{comment}

\small
\begin{table}[t]
\centering
\caption{Ordered sequential-passage probability $v_0(0,x_0) = \mathbb{P}^{x_0}(\tau_1 \le \tau_2 \le T)$
for the vehicle example.}
\label{tab:vehicle_validation}
\begin{tabular}{cccc}
\toprule
$v_{\max}$ & Theorem~\ref{thm:ordered} & Monte Carlo & rel.\ err. \\
\midrule
$1.2$ & 0.60219 & $0.60556 \pm 0.00303$ & $0.56\%$ \\
$1.3$ & 0.32321 & $0.32555 \pm 0.00290$ & $0.72\%$ \\
$1.4$ & 0.13094 & $0.13137 \pm 0.00209$ & $0.33\%$ \\
$1.5$ & 0.04089 & $0.04074 \pm 0.00123$ & $0.37\%$ \\
\bottomrule
\end{tabular}
\end{table}
\normalsize

\section{Conclusion and Future Work}

We have introduced probabilistic sequential safety, a notion that bounds
the probability that a stochastic process visits a tuple of closed target
regions by a deadline, either in a prescribed order or in any order.
Augmenting the process with a discrete mode or a visit set that records
sequential progress yields a strong Markov process. The Feynman--Kac
representation then characterises the safety probability through
recursively coupled backward Kolmogorov equations: $N$ for the ordered case
and $2^N-1$ for the unordered case. A longitudinal vehicle example
illustrates the results. Future work includes efficient data-driven
controller synthesis for sequential safety specifications.

 \bibliographystyle{IEEEtran}
 %\bibliography{lcss_ref}
 \bibliography{lcss_ref_updated_arxiv}
\end{document}